\documentclass{amsart}

\usepackage{amsmath,amssymb, amsthm}

\newcommand{\MC}{\mathbf{C}}
\newcommand{\AS}{\mathcal{S}}
\newcommand{\K}{\mathcal{Q}}
\newcommand{\Z}{\mathcal{P}}

\newcommand{\ZC}{\mathcal S^*_\tl G}

\newcommand{\FF}{\mathbb F}
\newcommand{\GG}{\mathcal G}

\newcommand{\tl}[1]{\mathbf {#1}}

\newcommand{\hor}[1]{^{(#1)}}
\newcommand{\uv}[1]{ ``#1'' }
\DeclareMathOperator{\del}{{\rm del}}

\DeclareMathOperator{\intdeg}{intdeg}
\DeclareMathOperator{\extdeg}{extdeg}
\DeclareMathOperator{\mdeg}{Mcdeg}

\newtheorem{corollary}{Corollary}
\newtheorem{theorem}{Theorem}
\theoremstyle{remark}
\newtheorem{remark}{Remark}

\begin{document}

\title{State spaces of convolutional codes, codings and encoders}

\author{\v St\v ep\' an Holub}

\email{holub@karlin.mff.cuni.cz}

\begin{abstract}
In this paper we give a compact presentation of the theory of abstract spaces for convolutional codes and convolutional encoders, and show a connection between them that seems to be missing in the literature. We use it for a short proof of two facts: the size of a convolutional encoder of a polynomial matrix is at least its inner degree, and the minimal encoder has the size of the external degree if the matrix is reduced. 
\end{abstract}

\maketitle

The theory of convolutional codes obtained a solid mathematical foundations with the work of David Forney \cite{Forney70, Forney73, Forney75}. The algebraic nature of the theory is the description of matrices over the field of rational sequences and/or over the ring of polynomials. Since then, any exposition of the theory consists of some mixture of this (linear) algebraic theory and more practical approach concerned with physical realization of linear time invariant systems. A classical example is the monograph by Johanesson and Zigangirov from 1999, recently republished in the second edition \cite{JohZig}. An important contribution represents the chapter by McEliece in Handbook of Coding Theory \cite{McEliece}, which features several relieving terminological reforms we adopt.
A more recent survey was given by Forney \cite{Forney11}.

One of the main ingredients of the algebraic approach is the concept of \uv{state space}, used already by Forney \cite{Forney70}, and developed by Willems in a wider context of dynamical systems (see \cite{Willems89} for a tutorial exposition). A closer look, however, leaves an impression that the theory of state spaces has so far not been explained in a way that fully exploits its strength. In particular, it seems that the relation between state space of the code and state space of the encoding is nowhere explicitly worked out in the literature. For example, the above book \cite{JohZig} does not define the state space of the code at all. On the other hand, McEliece's chapter uses mainly the state space of the code and speaks about the encoding state space only with respect to the abstract encoder (defined by four matrices over the basic field), not with respect to the generating matrix. 

In this paper we give a compact presentation of the theory of abstract spaces, and point out a simple, yet powerful, connection between them. In order to illustrate that we don't just juggle with terminology, we present a short proof of the fact that the size of a convolutional encoder of a polynomial matrix is at least the inner degree of the matrix (Corollary \ref{col}), which is considered by McEliece ``too long to include'' \cite[p. 1095]{McEliece}, although his exposition contains almost all tools needed for the proof.

 Using the same formulation as McEliece in the Preface to his text, we can say that this paper intends to add even more \uv{spin} to the theory introduced by Forney.  

\section{Notation}
Let $\FF$ be a (finite) field. An element $\tl u$ of the field $\FF((D))$ of {Laurent series} over $\FF$ is a series $\tl u = \sum_{j=z}^\infty u_j D^j$ over a formal variable $D$, where $z > -\infty$ is the least integer such that $u_z\neq 0$ called the \emph{delay} of $\tl u$, denoted by $\del \tl u$. The \emph{degree} of $\tl u$, denoted by $\deg \tl u$, is the largest index $j$ for which $u_j\neq 0$ (possibly $\infty$). (Note that $\del \tl 0=\infty$ and $\deg \tl 0=-\infty$.) We  avoid the tedious use of the variable $D$ in our notation: instead of $\tl u(D)$, the fact that we deal with a series is indicated by the boldface. Important subsets of $\FF((D))$ are the ring of polynomials $\FF[D]$, and its fraction field, the set of rational series $\FF(D)$. Let $\K(\tl u)=\sum_{j=0}^\infty u_j D^j$ be the \emph{causal} part of $\tl u$ and $\Z(\tl u)=\sum_{j=\del \tl u}^{-1} u_j D^j$ its \emph{anticausal} part. The set of causal series is the ring of power series $\FF[[D]]$. 
Bold face variables will also denote vectors of series. If  $\tl u \in \FF((D))^n$ or $\tl u \in \FF(D)^n$ , then $\tl u = \left(\tl u\hor 1, \tl u \hor 2,\dots,\tl u \hor n\right)$. 

\section{Levels of abstractness, state spaces and minimality}\label{minimality}

An $(n,k)$ convolutional code can be seen on three levels of increasing definiteness and/or decreasing abstractness: a \emph{code}, an \emph{encoding} and an \emph{encoder}.

1. The first level, called simply a (convolutional) \emph{code}, is the set $\MC$  of codewords, which is  by definition a $k$-dimensional subspace of the vector space $\FF((D))^n$ with $k$ dimensional subspace in $\FF(D)^n$, that is, generated by a basis from $\FF(D)^n$.

2. An \emph{encoding} is a \emph{realizable} linear mapping $\FF((D))^n \to \MC$ which assigns codewords to messages. Such an encoding is defined by the choice of a causal rational basis  of the space $\MC$ or, equivalently, by a \emph{generating  matrix} $\tl G$, which is a $k\times n$ matrix of rank $k$
whose rows
$\tl g_1,\tl g_2,\dots, \tl g_k \in \FF(D)^n$ generate $\MC$, and $\tl g_{i}\hor j$ are causal rational series. An information vector  $\tl u\in \FF((D))^k$ (a message) is mapped to the codeword $\tl u \tl G\in \FF((D))^n$.    

3. The encoding can be realized in real-time by an \emph{encoder} which is a transducer characterized by a finite vector space (over $\FF$)  of states $\AS$ and by a linear mapping $K:\AS\times \FF^k \to \AS\times \FF^n$. Given a state $s$ and an input $\vec u\in \FF^k$, the encoder outputs $\vec v\in \FF^n$ and enters a state $s'$, where $(s',\vec v)=K(s,\vec u)$. Since $K$ is linear, it can be expressed as a matrix of size $(\dim \AS + k)\times (\dim \AS + n)$ over $\FF$, which is composed of four matrices $\mathcal A,\mathcal  B,\mathcal  C, \mathcal  D$ such that $s' = s\, \mathcal  A + \vec u\, \mathcal  B$ and $v = s\, \mathcal  C + \vec u\, \mathcal D$. 
The set of states is then $\AS=\FF^m$, where $m=\dim \AS$. 
The encoder is typically realized by a linear circuit  which is  built from adders, multipliers and memory (or delay) elements. Then $m$ is the number of delay elements, and a state $s\in \FF^m$ is the content of encoder's memory.
\medskip

The main point is that each of the above three levels can be assigned a finite dimensional vector space of \emph{states} whose dimension is the corresponding \emph{degree}. We can therefore distinguish (with growing abstractness)  \emph{encoder state space}, \emph{encoding state space} and \emph{code state space}, as well as \emph{encoder degree}, \emph{encoding degree} and \emph{code degree}. 
\medskip

\begin{remark}
In the literature, terminology varies. The circuit realization is called a \uv{physical encoder} by McEliece \cite{McEliece}, or a \uv{physical realization} of \uv{abstract linear encoder}, where the latter is identified with the generating matrix, that is, with the encoding. McEliece points out that a matrix can have different physical realizations (even minimal ones). At the same time, however, the term \uv{convolutional encoder} without  qualification is used for the encoder defined by four matrices as above.  Those four matrices are also called \uv{state-space realization} of the matrix $\tl G$, this time without pointing out that the same matrix can have different state-space realizations. Johanesson and Zigangirov use the term \uv{encoder} directly for the circuit. 
\end{remark}
\medskip

  The state space of the encoder is explicit in its definition.
 On the other hand, it is not obvious what a \uv{state of an encoding}, or  \uv{state of a matrix $\tl G$} should be. In order to define the encoding state space, the encoding should be seen as the \uv{behaviour} of an encoder. The encoding state at a given instant of time is therefore naturally defined by the mapping of future inputs to future outputs. Since we deal with time invariant systems, it is convenient to denote the \uv{given instant} as time zero, and the state is then the mapping $\FF^k[[D]] \to \FF^n[[D]]$ of causal series. We are interested only in \uv{reachable} states, that is, those corresponding to some anticausal input $\tl v$.   
 Due to linearity, we have that for the causal input $\tl w\in \FF[[D]]^k$ the causal output will be 
\[\K(\tl v\tl G)+\tl w \tl G,\] 
which is the behaviour of the initial state modified by the addition of $\K(\tl v\tl G)$.
This suggests that we can define the \emph{encoding state space} $\Sigma_{\tl G}$ as the factor space
\[\Sigma_{\tl G}:=\FF((D))/\ZC\,,\]
where
\[\ZC=\{\tl u \in \FF((D)) \mid \K(\Z(\tl u) \tl G)=\tl 0\}\,.\]  
We denote elements of $\Sigma_\tl G$ by $[\tl u]_\tl G:= \tl u +\ZC$. Note, in particular, that $[\tl u]_{\tl G}= [\Z(\tl u)]_{\tl G}$ for any $\tl u \in \FF((D))$.
The set $\ZC$ is a vector space over the basic field $\FF$ (rather than over the field $\FF(D)$ of rational functions). Consequently, also $\Sigma_{\tl G}$ and $\FF((D))$ in the above definition are understood as vector spaces over $\FF$.  
Our definition of $\Sigma_{\tl G}$ may seem less intuitive than the 
definition used by Forney \cite{Forney70}, and Johanesson and Zigangirov \cite{JohZig}:
\[\Sigma_{\tl G}':=\{\K(\Z(\tl u)\tl G) \mid \tl u\in \FF((D))\}\,.\]
However, $\Sigma_{\tl G}$ and $\Sigma_{\tl G}'$ are isomorphic via the isomorphism $[\tl u]_\tl G\mapsto \K(\Z(\tl u)\tl G)$. The advantage of our definition will become clear with the definition of the mapping $\GG$ below.

One more step of the abstraction  is needed to introduce code state space which should be independent even  of a particular encoding.
The idea is illustrated  by the well-known construction of the minimal automaton in the theory of regular languages, when states are given by classes of the Myhill-Nerode equivalence, and thus defined purely in terms of the accepted language. In our case, it is natural to define \uv{code state zero} as the set 
\[\MC^*=\{\tl v \in \MC \mid \K(\tl v)\in \MC\}\]     
of causal code words,  again a vector space over  $\FF$. The \emph{state space} $\Sigma_\MC$ of $\MC$ is now defined as the factor space 
\[\Sigma_\MC:=\MC/\MC^*\,,\]
and the elements of $\Sigma_\MC$ are denoted by  $[\tl v]_\MC:= \tl v +\MC^*$. 

\begin{remark}\label{rem2}
This definition (and notation) is used by McEliece \cite[p. 1096]{McEliece} who uses the term \uv{abstract state space}. 
The same definition is used also in Forney \cite{Forney11}, where instead of $\MC^*$, the notation $\MC_- \times \MC_+$ is used (moreover indexed by the particular time instant for systems which may not be time invariant). 
Beware that \uv{abstract state space} denotes $\Sigma_{\tl G}'$ in Forney \cite{Forney70} as well as in Johanesson and Zigangirov \cite[p. 90]{JohZig}.
We remark that the ``behavioral'' approach leads to a slightly different definition 
\[\Sigma_\MC':=\K(\MC)/\MC_\K\,,\]
where
\begin{align*}
\K(\MC)&=\{\K(\tl u) \mid \tl u \in\MC\}, & \MC_\K&=\{\tl u \in \MC \mid \tl u=\K(\tl u)\}\,.
\end{align*}
This is the definition used (somewhat implicitly, see our concluding remarks) by Forney \cite{Forney70}.  
However, it is straightforward that $\K$ induces an isomorphism between $\Sigma_\MC$ and $\Sigma_\MC'$, which is also noted in Forney \cite{Forney11}.
\end{remark}

\medskip
An encoder is called \emph{minimal} if it has the smallest possible degree among all encoders realizing a given encoding. The degree of the minimal encoder of a generating matrix $\tl G$ is called the \emph{McMillan degree} of $\tl G$ (\cite[p. 1095]{McEliece}), denoted $\mdeg \tl G$. 
Since the state of an encoding is given by the (future) behavior of the encoder, it is obvious that the degree of the minimal encoder is at least the degree of the corresponding encoding.
On the other hand, the state space $\Sigma_{\tl G}$ can be used as a state space of an encoder $K_{\tl G}$ which could be called \uv{standard}. If $s = [\tl u]_{\tl G}$,  then $(s',\vec v)=K_{\tl G}(s, \vec u)$ where $s'= [\tl u D^{-1} + \vec u]_{\tl G}$ and $\vec v$ is the absolute component of $(\Z(\tl u) + \vec u)\tl G$. It is easy to see that $K_{\tl G}$ is linear and well defined. 
Therefore, the McMillan degree of $\tl G$ is in fact the dimension of the space $\Sigma_{\tl G}$: $$\deg \Sigma_{\tl G}= \mdeg \tl G\,.$$

Similarly, an encoding (a generating matrix $\tl G$) is called \emph{minimal} if $\dim \Sigma_{\tl G}$ is the smallest possible among all encodings defining the same code. A simple tool for studying minimal encodings is the linear mapping 
\begin{align*}
\GG: \Sigma_{\tl G}&\to \Sigma_\MC \\
     [\tl u]_{\tl G}&\mapsto [\tl u \tl G]_{\MC}. 
\end{align*}
Let us verify that $\GG$ is well defined, which means that $\tl G$ maps $\ZC$ into $\MC^*$. For $\tl u\in \ZC$, we have
\[\K(\tl u\tl G)=\K(\Z(\tl u)\tl G) + \K(\K(\tl u)\tl G)=\K(\tl u)\tl G \]
and therefore $\tl u \tl G$ is indeed in $\MC^*$.
 We now have 
 \[
 \dim \Sigma_\MC \leq \dim \Sigma_{\tl G}\,,
 \]
 and $\tl G$ is minimal if $\GG$ is injective.
This directly translates into the equivalent criterion ``Only the zero abstract state of $\tl G$ can be a codeword'' (see Theorem 2.37 (iii) \cite{JohZig}). 

\begin{remark}
An encoder is a minimal possible encoder of a code if it is a minimal encoder of a minimal encoding. However, the minimal possible encoder can be constructed directly from $\Sigma_{\MC}$ in a way similar to the above construction of the \uv{standard encoder}. Such a minimal realization was studied by Willems and is called \uv{canonical realization} by Forney \cite{Forney11}.   	
\end{remark}

 \section{Polynomial generating matrices}
Consider a polynomial generating matrix $\tl G$. Define the degree of the $i$th row as $$\deg \tl g_i= \max_j \deg \tl g_i\hor j,$$ and the \emph{external degree} of $\tl G$ as $\extdeg \tl G=\sum_{i=1}^k \deg \tl g_i$. Assume, without loss of generality, that $\deg \tl g_1\leq \deg \tl g_2\leq \cdots \leq \deg \tl g_k$.  It is well known (and not difficult to see) that, for a given code $\MC$, $\extdeg \tl G$ is minimized if and only if $\deg \tl g_i=e_i$ for a uniquely given $k$-tuple $(e_1,e_2,\dots, e_k)$ of \emph{Forney indeces} of $\MC$. Then $\tl G$ is called a \emph{canonical} generating matrix of $\MC$.

Each polynomial matrix $\tl G$ admits a circuit realization whose degree is $\extdeg \tl G$ (this is the most obvious realization called \uv{direct-form}  or \uv{controller canonical form} \cite{McEliece, JohZig}). Using Section \ref{minimality}, we therefore have 
\[ \dim \Sigma_{\MC} \leq \mdeg \tl G \leq \extdeg \tl G\,. \]

One of the fundamental results of the Forney theory is the  State Space Theorem, which claims that  \[ \dim \Sigma_\MC = \sum_{i=1}^k e_i.\] This means that canonical matrices define minimal encodings of a given code, and moreover, also their direct-form realization is the minimal one.

A standard characterization of canonical matrices is the following: A polynomial matrix is canonical if and only if it is \emph{basic}, and \emph{reduced}. Each of these two properties is captured by several equivalent conditions. 

I. Being reduced is related to the external degree of the matrix. A matrix $\tl G$ is reduced if its external degree cannot be lowered by adding polynomial multiples of some rows to another row. Equivalently, a reduced matrix of $\MC$ is obtained from a polynomial generating matrix by multiplication by a unimodular matrix (polynomial matrix with the determinant in $\FF$) minimizing the external degree. The matrix $\tl G$ is reduced if and only if it has \emph{the predictable degree property} which states that
\[\deg \tl u \tl G = \max_{i=1,\dots,k} \deg \tl u\hor i \tl g_i\]
for each $\tl u\in \FF((D))^k$.

II. Being basic is related to the \emph{internal degree}  of the matrix, which is the largest degree of its $k\times k$ minors. A polynomial generating matrix $\tl G$ of $\MC$ is basic if its internal degree $\intdeg \tl G$ is minimal among all polynomial generating matrices of $\MC$.  Another equivalent characterization is that $\tl G$ is basic if and only if it has a polynomial right inverse.

It is also not difficult to see that $\intdeg \tl G\leq \extdeg \tl G$ for any polynomial matrix, and that $\intdeg \tl G= \extdeg \tl G$ if and only if $\tl G$ is reduced.

The above claims can be found in any textbook about convolutional codes (see \cite[Appendix A]{McEliece} for a longer list of equivalent conditions) and we will use them in what follows.

\medskip

We formulate the State Space Theorem in a way that reveals the main idea of the proof, which closely follows McEliece \cite{McEliece}. The idea is to consider a particular basis of $\Sigma_{\MC}$ defined by the canonical matrix $\tl G$.   
\begin{theorem}\label{statespace}
Let $\tl G$ be a canonical generating matrix of $\MC$. Then \[\{[D^{-j}\tl g_i]_\MC \mid i=1,\dots,k, \ j=1,\dots,e_i\}\] is a basis of $\Sigma_\MC$.
\end{theorem}
\begin{proof}
	Take an arbitrary element of $\MC$ and write it as $\tl u \tl G$. Decompose $\tl u=\tl u_1 + \tl u_2 + \tl u_3$ so that
	\begin{align*}
\deg \tl u_1\hor i & < -e_i,  \\
\del \tl u_2\hor i &\geq -e_i, \quad \deg \tl u_2\hor i<0,\\
\del \tl u_3\hor i &\geq 0\,.   
	\end{align*}
We have $\K(\tl u_1\tl G)=\tl 0\in \MC^*$ and $\K(\tl u_3\tl G)=\tl u_3\tl G\in \MC^*$, 
which implies $[\tl u\tl G]_\MC=[\tl u_2\tl G]_\MC$. The vector $\tl u_2\tl G$ is a linear combination of vectors $D^{-j}\tl g_i$ (over $\FF$), 
namely
\[
\tl u_2\tl G = \sum_{i=1}^b \sum_{j=1}^{e_i} u_j\hor i D^{-j}\tl g_i\,.
\]
We have shown that vectors $[D^{-j}\tl g_i]_\MC$ generate $\Sigma_\MC$.
	
It remains to show linear independence. Let therefore
\[\sum_{i=1}^b \sum_{j=1}^{e_i} u_j\hor i  D^{-j}\tl g_i \in \MC^*. \]  
This can be written as
\[\K(\tl u \tl G) = \tl w \tl G\]
for some $\tl w$, where $$\tl u\hor i = \sum_{j=-e_j}^{-1} u_j\hor i D^{j}.$$ 
Since $\tl G$ has a polynomial right inverse $\tl G'$, we have $\tl w=\K(\tl u \tl G)\tl G'$, hence $\tl w$ is polynomial. Then $\deg (\tl u - \tl w)\hor i \tl g_i \geq 0$, unless $(\tl u - \tl w)\hor i = \tl 0$. On the other hand, we have $\K((\tl u-\tl w)\tl G)=\tl 0$, which means that $\deg (\tl u-\tl w)\tl G < 0$. 
The predictable degree property of $\tl G$ now implies that $\tl u-\tl w=\tl 0$, therefore also $\tl u=\tl w=\tl 0$, since $\tl u$ is anticausal and $\tl w$ is causal. This completes the proof.
\
\end{proof}

We now turn our attention to
 bases of $\Sigma_{\tl G}$. For any reduced matrix $\tl G$, there is a characterization of a basis of $\Sigma_{\tl G}$ which is quite analogical to the State Space Theorem above. Nevertheless, while the State Space Theorem has become a classical result, a clear exposition of the following parallel result seems to be missing in the literature.

We denote by $\tl e_i$ the $i$th canonical vector (which should be not confused with the $i$th Forney index $e_i$).
\begin{theorem} \label{reduced}
Let $\tl G$ be a reduced generator matrix with $\deg \tl g_i = \nu_i$, $i=1,2,\dots,k$. Then the set
 \[\{[D^{-j}\tl e_i]_\tl G \mid i=1,\dots,k, \ j=1,\dots,\nu_i\}\] 
is a basis of $\Sigma_\tl G$.
\end{theorem}
\begin{proof}
The proof closely resembles the proof of Theorem \ref{statespace}. 
Let  $\tl u=\tl u_1 + \tl u_2 + \tl u_3$ be a decomposition of an arbitrary $\tl u\in \FF((D))$ given by 
	\begin{align*}
\deg \tl u_1\hor i & < -\nu_i,  \\
\del \tl u_2\hor i &\geq -\nu_i, \quad \deg \tl u_2\hor i<0,\\
\del \tl u_3\hor i &\geq 0\,.   
	\end{align*}
 From the definition of $\ZC$, we obtain $[\tl u]_\tl G=[\tl u_2]_\tl G$ which implies that vectors $[D^{-j}\tl e_i]_\tl G$ generate $\Sigma_\tl G$. 

To show linear independence, consider $\tl u\in \ZC$ where 
\[\tl u\hor i=\sum_{j=-\nu_i}^{-1} u_j\hor i D^{j}\tl e_i. \]
Since $\tl u=\Z(\tl u)$, we deduce from  $\tl u\in \ZC$ that $\K(\tl u \tl G) = \tl 0$, and the predictable degree property of $\tl G$ implies $\tl u=\tl 0$.
\end{proof}

For general polynomial generating matrices, a variant of the linear independence part of the proof can be carried out, providing a lower bound in terms of an equivalent reduced matrix.

\begin{theorem}\label{notr}
Let $\tl G$ be a polynomial generator matrix and let $\tl T$ be a unimodular matrix such that $\tl T\tl G$ is a reduced matrix with row degrees $(\nu_1,\dots,\nu_k)$. Then the set
 \[\{[D^{-j}\tl t_i]_\tl G \mid i=1,\dots,k, \ j=1,\dots,\nu_i\},\] 
where $\tl t_i$ are rows of $\tl T$, is linearly independent subset of $\Sigma_{\tl G}$.
\end{theorem}
\begin{proof}
	Let $\tl u'\in \ZC$ for some 
\[\tl u'=\sum_{i=1}^b \sum_{j=1}^{\nu_i} u_j\hor i D^{-j}\tl t_i. \]
Then $\tl u'=\tl u \tl T$, where
\[\tl u\hor i=\sum_{j=-\nu_i}^{-1} u_j\hor i D^{j},  \]
 and 
\begin{align*}
\tl 0 & =\K(\Z(\tl u \tl T)\tl G)= \K((\tl u\tl T - \K(\tl u\tl T)) \tl G)\\ &=\K(\tl u \tl T \tl G) - \K(\tl u\tl T) \tl G.
\end{align*}

From this and from $\tl u=\Z(\tl u)$ we obtain that $\K((\Z(\tl u) - \tl w) \tl T \tl G) = \tl 0$, where $\tl w=\K(\tl u\tl T)\tl T^{-1}$. Since both $\tl T$ and $\tl T^{-1}$ are polynomial, we deduce that also $\tl w$ is polynomial. The predictable degree property of $\tl T\tl G$ implies that $\Z(\tl u) - \tl w=\tl 0$, and therefore also $\tl u=\tl 0$, and $\tl u'=\tl 0$.  
\end{proof}

\begin{corollary}\label{col}
For any polynomial matrix $\tl G$,
\[\intdeg \tl G\leq\mdeg \tl G \,.\]
If $\tl G$ is reduced, then \[\intdeg \tl G=\mdeg \tl G = \extdeg \tl G\,.\]
\end{corollary}
\begin{proof}
 Since $\tl T$ is unimodular and $\tl T \tl G$ is reduced, Theorem \ref{notr} yields
	\begin{align*}
\intdeg \tl G &= \intdeg \tl T\tl G =\extdeg \tl T\tl G \\ 
              &\leq \dim \Sigma_{\tl G} \leq \extdeg \tl G\,.
	\end{align*}
The claim follows.	
\end{proof}

\section{Concluding remarks}
We have shown that it is useful to clearly distinguish  the state space of the code from the state space of the encoding, and to consider the relation between the two. Although the relation is as simple as the mapping $\GG$, it has apparently not been used in the literature so far.

A historical reason for this may be that  Forney uses the factor space $\K(\MC)/\FF[[D]]\tl G$ in his seminal paper \cite[Lemma 5]{Forney70} (denoted there as $S_G=CQ/C_G$). This is equivalent to $\K(\MC)/\MC_\K$ if and only if $\tl G$ has a causal right inverse. However, for a general matrix this space is not very useful. In particular, unlike $\Sigma_\MC$, the space $\K(\MC)/\FF[[D]]\tl G$ is not independent of $\tl G$.

McEliece \cite{McEliece} uses explicitly the space $\MC/\MC^*$, and he calls it the \uv{abstract state space} which is an important terminological shift with respect to the usage introduced by Forney \cite{Forney70} (cf. Remark \ref{rem2} above) not pointed out in the \uv{Advice to experts} \cite[Preface, p.1068]{McEliece}. Nevertheless, not even McEliece liberates code degree from a particular encoding, and (following Forney \cite{Forney70}) defines $\deg \MC$ as the minimal possible external degree of a polynomial generating matrix. Such a definition, however, is fully justified only when the State Space Theorem 
is proven.

The present paper can be summarized as a vindication of a different approach: if, as we did, $\deg \MC$ is \emph{defined} as $\dim \Sigma_\MC$, then the mapping $\GG$ immediately shows that we cannot hope for any smaller encoding, let alone for a smaller encoder. The State Space Theorem then shows that this absolute minimum is not only achievable, it is even achieved by direct-form encoders of (canonical) polynomial matrices.

\bibliographystyle{IEEEtran}
\bibliography{IEEEabrv,abstract}

\end{document}